\documentclass[preprint]{imsart}

\RequirePackage{amsthm,amsmath,amsfonts,amssymb}
\RequirePackage[numbers]{natbib}
\RequirePackage[colorlinks,citecolor=blue,urlcolor=blue]{hyperref}
\RequirePackage{graphicx}

\startlocaldefs
\endlocaldefs

\newtheorem{thm}{Theorem}
\newtheorem{lem}[thm]{Lemma}

\theoremstyle{definition}

\newtheorem{ex}[thm]{Example}
\newtheorem{rem}[thm]{Remark}

\providecommand{\abs}[1]{\lvert#1\rvert}

\begin{document}

\begin{frontmatter}
\title{A probabilistic view on predictive constructions for Bayesian learning}
%\title{A sample article title with some additional note\thanksref{t1}}
\runtitle{Predictive approach to Bayesian learning}
%\thankstext{T1}{A sample additional note to the title.}

\begin{aug}
%%%%%%%%%%%%%%%%%%%%%%%%%%%%%%%%%%%%%%%%%%%%%%
%%Only one address is permitted per author. %%
%%Only division, organization and e-mail is %%
%%included in the address.                  %%
%%Additional information can be included in %%
%%the Acknowledgments section if necessary. %%
%%%%%%%%%%%%%%%%%%%%%%%%%%%%%%%%%%%%%%%%%%%%%%
\author[A]{\fnms{Patrizia} \snm{Berti}\ead[label=e1]{patrizia.berti@unimore.it}}
\author[B]{\fnms{Emanuela} \snm{Dreassi}\ead[label=e2]{emanuela.dreassi@unifi.it}}
\author[C]{\fnms{Fabrizio} \snm{Leisen}\ead[label=e3]{fabrizio.leisen@gmail.com}}
\author[D]{\fnms{Luca} \snm{Pratelli}\ead[label=e4]{pratel@mail.dm.unipi.it}}
\and
\author[E]{\fnms{Pietro} \snm{Rigo}\ead[label=e5]{pietro.rigo@unibo.it}}
%%%%%%%%%%%%%%%%%%%%%%%%%%%%%%%%%%%%%%%%%%%%%%
%% Addresses                                %%
%%%%%%%%%%%%%%%%%%%%%%%%%%%%%%%%%%%%%%%%%%%%%%
\address[A]{Dipartimento di Scienze Fisiche, Informatiche e Matematiche, Universit\`a di Modena e Reggio-Emilia, via Campi 213/B, 41100 Modena, Italy
\printead{e1}}

\address[B]{Dipartimento di Statistica, Informatica, Applicazioni, Universit\`a di Firenze, viale Morgagni 59, 50134 Firenze, Italy
\printead{e2}}

\address[C]{School of Mathematical Sciences, University of Nottingham, University Park, Nottingham, NG7 2RD, UK
\printead{e3}}

\address[D]{Accademia Navale, viale Italia 72, 57100 Livorno,
Italy
\printead{e4}}

\address[E]{Dipartimento di Scienze Statistiche ``P. Fortunati'', Universit\`a di Bologna, via delle Belle Arti 41, 40126 Bologna, Italy
\printead{e5}}
\end{aug}

\begin{abstract}
Given a sequence $X=(X_1,X_2,\ldots)$ of random observations, a Bayesian forecaster aims to predict $X_{n+1}$ based on $(X_1,\ldots,X_n)$ for each $n\ge 0$. To this end, in principle, she only needs to select a collection $\sigma=(\sigma_0,\sigma_1,\ldots)$, called ``strategy" in what follows, where $\sigma_0(\cdot)=P(X_1\in\cdot)$ is the marginal distribution of $X_1$ and $\sigma_n(\cdot)=P(X_{n+1}\in\cdot\mid X_1,\ldots,X_n)$ the $n$-th predictive distribution. Because of the Ionescu-Tulcea theorem, $\sigma$ can be assigned directly, without passing through the usual prior/posterior scheme. One main advantage is that no prior probability is to be selected. In a nutshell, this is the predictive approach to Bayesian learning. A concise review of the latter is provided in this paper. We try to put such an approach in the right framework, to make clear a few misunderstandings, and to provide a unifying view. Some recent results are discussed as well. In addition, some new strategies are introduced and the corresponding distribution of the data sequence $X$ is determined. The strategies concern generalized P\'olya urns, random change points, covariates and stationary sequences.
\end{abstract}



\end{frontmatter}

\maketitle

\section{Introduction}\label{bt6yh}

This paper has been written having the following interpretation of Bayesian inference in mind. (We declare this interpretation from the outset just to make transparent our point of view and easier the understanding of the paper). Let us call $\mathcal{O}$ the object of inference. Roughly speaking, $\mathcal{O}$ denotes whatever we ignore but would like to know. For instance, $\mathcal{O}$ could be a parameter (finite or infinite dimensional), a set of future observations, an unknown probability distribution, the effect of some action, or something else. According to us, the distinguishing feature of the Bayesian approach is to regard $\mathcal{O}$ as the realization of a random element, and not as an unknown but fixed constant. As a consequence, the main goal of any Bayesian inferential procedure is to determine the conditional distribution of $\mathcal{O}$ given the available information.

\medskip

Note that, unless $\mathcal{O}$ itself is a parameter, no other parameter is necessarily involved.

\medskip

Prediction of unknown observable quantities is a fundamental part of statistics. Initially, it was probably the most prevalent form of statistical inference. The wind changed at the beginning of the 20$^{\text{th}}$ century when statisticians' attention shifted to other issues, such as parametric estimation and testing; see e.g. \cite{GEISSER}. Nowadays, prediction is back in the limelight again, and plays a role in modern topics including machine learning and data mining; see e.g. [\citealp{CFZ}, \citealp{BCJC}, \citealp{EFRON}, \citealp{HTF}].

\medskip

This paper deals with prediction of future observations, based on the past ones, from the Bayesian point of view. Precisely, we focus on a sequence
$$X=(X_1,X_2,\ldots)$$
of random observations and, at each time $n$, we aim to predict $X_{n+1}$ based on $(X_1,\ldots,X_n)$. Hence, for each $n$, the object of inference is $\mathcal{O}=X_{n+1}$, the available information is $(X_1,\ldots,X_n)$, and the target is the {\em predictive distribution} $P(X_{n+1}\in\cdot\mid X_1,\ldots,X_n)$. We point out that, apart from technicalities, most of our considerations could be generalized to the case where $\mathcal{O}$ is an arbitrary (measurable) function of the future observations, say
$$\mathcal{O}=f(X_{n+1},X_{n+2},\ldots).$$
This case is recently object of increasing attention; see e.g. [\citealp{FHW2021}, \citealp{HSOLO}].

\medskip

No parameter $\theta$ plays a role at this stage. The forecaster may involve some $\theta$, if she thinks it helps, but she is not interested in $\theta$ as such. To involve $\theta$ means to model the probability distribution of $X$ as depending on $\theta$, and then to exploit this fact to calculate the predictive distributions $P(X_{n+1}\in\cdot\mid X_1,\ldots,X_n)$.

\medskip

To better address our prediction problem, it is convenient to introduce the notion of {\em strategy}. Let $(S,\mathcal{B})$ be a measurable space, with $S$ to be viewed as the set where the observations $X_n$ take values. Following Dubins and Savage \cite{DS1965}, a strategy is a sequence
$$\sigma=(\sigma_0,\sigma_1,\sigma_2,\ldots)$$
such that

\begin{itemize}

\item $\sigma_0$ and $\sigma_n(x)$ are probability measures on $\mathcal{B}$ for all $n\ge 1$ and $x\in S^n$;

\item The map $x\mapsto\sigma_n(x,A)$ is $\mathcal{B}^n$-measurable for fixed $n\ge 1$ and $A\in\mathcal{B}$.

\end{itemize}

Here, $\sigma_0$ should be regarded as the marginal distribution of $X_1$ and $\sigma_n(x)$ as the conditional distribution of $X_{n+1}$ given that $(X_1,\ldots,X_n)=x$. Moreover, $\sigma_n(x,A)$ denotes the value taken at $A$ by the probability measure $\sigma_n(x)$. We also note that strategies are often called {\em prediction rules} in the framework of species sampling sequences; see \cite[p. 251]{PIT1996}.

\medskip

Strategies are a natural tool to frame a prediction problem from the Bayesian standpoint. In fact, a strategy $\sigma$ can be regarded as the collection of all predictive distributions (including the marginal distribution of $X_1$) in the sense that $\sigma_n(x,\,\cdot)=P\bigl(X_{n+1}\in\cdot\mid (X_1,\ldots,X_n)=x\bigr)$ for all $n\ge 0$ and $x\in S^n$. Thus, in a sense, everything a Bayesian forecaster has to do is to select a strategy $\sigma$. Obviously, the problem is how to do it. A related problem is whether, in order to choose $\sigma$, involving a parameter $\theta$ is convenient or not.

\medskip

An important special case is exchangeability. In fact, if $X$ is assumed to be exchangeable, there is natural way to involve a parameter $\theta$. To see this, take the parameter space $\Theta$ as
$$\Theta=\bigl\{\text{all probability measures on }\mathcal{B}\bigr\}.$$
Moreover, for each $\theta\in\Theta$, denote by $P_\theta$ a probability measure which makes $X$ i.i.d. with common distribution $\theta$, i.e.,
$$
P_\theta\bigl(X_1\in A_1,\ldots,X_n\in A_n\bigr)=\prod_{i=1}^n\theta(A_i)
$$
for all $n\ge 1$ and $A_1,\ldots,A_n\in\mathcal{B}$. Then, under mild conditions on $(S,\mathcal{B})$, de Finetti's theorem yields
\begin{gather*}
P(X\in\cdot)=\int_\Theta P_\theta(X\in\cdot)\,\pi(d\theta)
\end{gather*}
for some (unique) prior probability $\pi$ on $\Theta$. Thus, conditionally on $\theta\in\Theta$, the observations are i.i.d. with common distribution $\theta$. This suggests calculating the strategy $\sigma$ as follows.

\medskip

\begin{itemize}

\item[(i)] Select a prior $\pi$ on $\Theta$;

\item[(ii)] For each $n\ge 1$ and $x\in S^n$, evaluate the posterior of $\theta$ given $x$, namely, the conditional distribution of $\theta$ given that $(X_1,\ldots,X_n)=x$;

\item[(iii)] Calculate $\sigma$ as
$$\sigma_n(x,A)=\int_\Theta\theta(A)\,\pi_n(d\theta\mid x)\quad\quad\text{for all }A\in\mathcal{B},$$
where $\pi_n(\cdot\mid x)$ is the posterior and $\pi_0(\cdot\mid x)$ is meant as $\pi_0(\cdot\mid x)=\pi$.
\end{itemize}

\medskip

Steps (i)-(ii)-(iii) are familiar in a Bayesian framework. Henceforth, if $\sigma$ is selected via (i)-(ii)-(iii), the forecaster is said to follow the {\em inferential approach} (I.A.).

\medskip

\subsection{Predictive approach to Bayesian modeling} There is another approach to Bayesian prediction, usually called the {\em predictive approach} (P.A.), which is quite recurrent in the Bayesian literature and recently gained increasing attention. (Such an approach, incidentally, has been referred to as the ``non-standard approach" in [\citealp{BDPR2021}, \citealp{BPCID}]). According to P.A., the forecaster directly selects her strategy $\sigma$. Merely, for each $n\ge 0$, she selects the predictive $\sigma_n$ without passing through the prior/posterior scheme described above. Among others, P.A. is supported by de Finetti, Savage, Dubins [\citealp{DEF1931}, \citealp{DFB1937}, \citealp{DS1965}] and more recently by Diaconis and Regazzini [\citealp{BRR1997}, \citealp{CR1996}, \citealp{DY1979}, \citealp{DF1990}, \citealp{FLR2000}]. P.A. is also strictly connected to Dawid's prequential approach [\citealp{DAW1984}, \citealp{DAW1992}, \citealp{DV1999}] and to Pitman's treatment of species sampling sequences [\citealp{PIT1996}, \citealp{PITYOR}, \citealp{PIT06}]. In addition, several prediction procedures arising in non-necessarily Bayesian frameworks, such as Machine Learning and Data Mining, are consistent with P.A.; see e.g. [\citealp{CFZ}, \citealp{BCJC}, \citealp{EFRON}, \citealp{HTF}]. Some further related references are [\citealp{BDPR2021}, \citealp{BPCID}, \citealp{FHW2021}, \citealp{FL2022}, \citealp{FP2012}, \citealp{HSOLO}, \citealp{HMW}, \citealp{HILL}].

\medskip

The theoretical foundation of P.A. is the Ionescu-Tulcea theorem; see e.g. \cite[p. 159]{HJ1994}. Roughly speaking this theorem states that, to assign the joint distribution of $X$, it suffices to choose, {\em in an arbitrary way}, the marginal  distribution of $X_1$, the conditional distribution of $X_2$ given $X_1$, the conditional distribution of $X_3$ given $(X_1,X_2)$, and so on. Note that this fact would be obvious if $X$ would be replaced by a finite dimensional random vector $(X_1,\ldots,X_m)$. So, in a sense, the Ionescu-Tulcea theorem extends to infinite sequences a straightforward property of finite dimensional vectors. In any case, a formal statement of the theorem is as follows.

\begin{thm}\label{tulc}\textbf{(Ionescu-Tulcea).}
For each $n\ge 1$, let $X_n$ be the $n$-th coordinate random variable on $(S^\infty,\mathcal{B}^\infty)$. Then, for any strategy $\sigma$, there is a unique probability measure $P_\sigma$ on $(S^\infty,\mathcal{B}^\infty)$ such that
\begin{gather}\label{itt}
P_\sigma(X_1\in\cdot)=\sigma_0(\cdot)\quad\quad\text{and}
\\P_\sigma\bigl(X_{n+1}\in\cdot\mid (X_1,\ldots,X_n)=x\bigr)=\sigma_n(x,\cdot)\notag
\end{gather}
for all $n\ge 1$ and $P_\sigma$-almost all $x\in S^n$.
\end{thm}

Because of Theorem \ref{tulc}, to make predictions on the sequence $X$, the forecaster is free to select an arbitrary strategy $\sigma$. In fact, for any $\sigma$, there is a (unique) probability distribution for $X$, denoted above by $P_\sigma$, whose predictives $P_\sigma\bigl(X_{n+1}\in\cdot\mid X_1,\ldots,X_n\bigr)$ agree with $\sigma$ in the sense of equation \eqref{itt}.

\medskip

The strengths and weaknesses of I.A. versus P.A. are discussed in a number of papers; see e.g. [\citealp{BDPR2021}, \citealp{BCJC}, \citealp{EFRON}, \citealp{GEISSER}, \citealp{GALIIT2010}] and references therein. Here, we summarize this issue (from our point of view) under the assumption that prediction is the main target.

\medskip

I.A. is not motivated by prediction alone. The main goal of I.A. is to make inference on other features of the data distribution (typically some parameters) and in this case the prior $\pi$ is fundamental. It should be added that $\pi$ often provides various meaningful information on the data generating process. However, to assess $\pi$ is not an easy task. In addition, once $\pi$ is selected, to evaluate the posterior $\pi_n(\cdot\mid x)$ is quite difficult as well. Frequently, $\pi_n(\cdot\mid x)$ cannot be written in closed form but only approximated numerically. In short, I.A. is a cornerstone of Bayesian inference, but, when prediction is the main target, it is actually quite involved.

\medskip

In turn, P.A. has essentially four merits. First, P.A. allows to avoid an explicit choice of the prior $\pi$. Indeed, when prediction is the main target, why select $\pi$ explicitly ? Rather than wondering about $\pi$, it seems reasonable to reflect on how the information in $(X_1,\ldots,X_n)$ is conveyed in the prediction of $X_{n+1}$. Second, the data sequence $X$ is not required any distributional assumption. This point is developed in Subsections \ref{xxz34} and \ref{zq219uj}. By now, we stress a consequence of such a point. The Bayesian nature of a prediction procedure does not depend on the data distribution. For instance, a forecaster applying P.A. is certainly Bayesian independently of the distribution attached to $X$. Third, P.A. requires the assignment of probabilities
on observable facts only. The value of $X_{n+1}$ is actually observable, while
$\pi$ and $\pi_n$ (being probabilities on $\Theta$) do not necessarily deal with observable facts. Fourth, the strategy $\sigma$ may be assigned stepwise. At each time $n$, the forecaster has observed $x=(x_1,\ldots,x_n)\in S^n$ and has already selected $\sigma_0,\sigma_1(x_1),\ldots,\sigma_{n-1}(x_1,\ldots,x_{n-1})$. Then, to predict $X_{n+1}$, she is still free to select $\sigma_n(x)$ as she wants. No choice of $\sigma_n(x)$ is precluded. This is consistent with the Bayesian view, where the observed data are fixed and one should condition on them. In spite of these advantages, P.A. has an obvious drawback. In fact, assigning a strategy $\sigma$ directly may be very difficult, in principle as difficult as selecting a prior $\pi$.

\medskip

A last (basic) remark is that, if $X$ is exchangeable, both I.A. and P.A. completely determine the probability distribution of $X$. Selecting a prior $\pi$ or choosing a strategy $\sigma$ are just equivalent routes to fix the distribution of $X$. In particular, selecting $\sigma$ uniquely determines $\pi$. An intriguing line of research is in fact to identify the prior corresponding to a given $\sigma$; see e.g. [\citealp{DIRGEN}, \citealp{DY1979}, \citealp{DF1990}, \citealp{FLR2000}].

\subsection{Characterizations}\label{xxz34}
Recall that, for any strategy $\sigma$, there is a unique probability measure $P_\sigma$ on $(S^\infty,\mathcal{B}^\infty)$ satisfying condition \eqref{itt}.

\medskip

In principle, when applying P.A., the data sequence $X$ is free to have any probability distribution. Nevertheless, in most applications, it is reasonable (if not mandatory) to impose some conditions on $X$. For instance, the forecaster may wish $X$ to be exchangeable, or stationary, or Markov, or a martingale, and so on. In these cases, $\sigma$ is subjected to some constraints. If $X$ is required to be exchangeable, for instance, $\sigma$ should be such that $P_\sigma$ is exchangeable. Hence, those strategies $\sigma$ which make $P_\sigma$ exchangeable should be characterized.

\medskip

More generally, fix any collection $\mathcal{C}$ of probability measures on $(S^\infty,\mathcal{B}^\infty)$ and suppose the data distribution is required to belong to $\mathcal{C}$. Then, P.A. gives rise to the following problem:
\begin{itemize}

\item[ ]\textbf{Problem (*):} Characterize those strategies $\sigma$ such that $P_\sigma\in\mathcal{C}$.

\end{itemize}
Sometimes, Problem (*) is trivial (Markov, martingales) but sometimes it is not (stationarity, exchangeability). To illustrate, we mention three examples (which correspond to the three dependence forms examined in the sequel).

\medskip

In the exchangeable case, Problem (*) admits a solution \cite[Th. 3.1]{FLR2000} but the conditions on $\sigma$ are quite hard to check in real problems. Hence, applying P.A. to exchangeable data is usually difficult (even if there are some exceptions; see Section \ref{d5tn9j}).

\medskip

A condition weaker than exchangeability is conditional identity in distribution. Say that $X$ is {\em conditionally identically distributed} (c.i.d.) if $X_2\overset{d}=X_1$ and, for each $n\ge 1$, the conditional distribution of $X_k$ given $(X_1,\ldots,X_n)$ is the same for all $k>n$; see Section \ref{xe5ty7}. It can be shown that
$$X\text{ is exchangeable}\quad\Leftrightarrow\quad X\text{ is stationary and c.i.d.;}$$
see [\citealp{BPR2004}, \citealp{K}]. Hence, conditional identity in distribution can be regarded as one of the two basic ingredients of exchangeability (the other being stationarity). Now, in the c.i.d. case, Problem (*) has been solved \cite[Th. 3.1]{BPR2012} and the conditions on $\sigma$ are quite simple. The class of admissible strategies includes several meaningful elements which cannot be used if $X$ is required to be exchangeable. As a consequence, P.A. works quite well for c.i.d. data; see [\citealp{BDPR2021}, \citealp{BPCID}].

\medskip

The stationary case is more involved. In fact, to our knowledge, there is no general characterization of the strategies $\sigma$ which make $P_\sigma$ stationary. However, such a characterization is available in some meaningful special cases (e.g. when $P_\sigma$ is also required to be Markov); see Section \ref{x76t}.

\medskip

Finally, Problem (*) is usually easier in a few (meaningful) special cases. For instance, Problem (*) is simpler if $P_\sigma$ is also asked to be Markov; see e.g. \cite{FPB2017} and Section \ref{x76t}. Or else, if the strategy $\sigma$ is required to be dominated.

\begin{itemize}
\item[ ]\textbf{Dominated strategies:} Let $\lambda$ be a $\sigma$-finite measure on $(S,\mathcal{B})$. Say that a strategy $\sigma$ is dominated by $\lambda$ if each $\sigma_n(x)$ admits a density $f_n(\cdot\mid x)$ with respect to $\lambda$, namely,
\begin{gather*}
\sigma_0(dy)=f_0(y)\,\lambda(dy)\quad\quad\text{and}
\\\sigma_n(x,\,dy)=f_n(y\mid x)\,\lambda(dy)
\end{gather*}
for all $n\ge 1$ and $x\in S^n$. Here, $f_0:S\rightarrow\mathbb{R}^+$ and $f_n:S\times S^n\rightarrow\mathbb{R}^+$ are non-negative measurable functions.
\end{itemize}

For instance, if $S=\mathbb{R}$ and $\sigma_n(x)$ is a non-degenerate normal distribution for all $n$ and $x$, then $\sigma$ is dominated by $\lambda=\,$Lebesgue measure. Or else, if $S$ is countable, any strategy is dominated by $\lambda=\,$counting measure. Instead, if $S$ is uncountable, a non-dominated strategy is $\sigma_n(x_1,\ldots,x_n)=\delta_{x_n}$ where $\delta_{x_n}$ denotes the unit mass at the point $x_n$. Another non-dominated strategy is the empirical measure $\sigma_n(x_1,\ldots,x_n)=(1/n)\,\sum_{i=1}^n\delta_{x_i}.$

\medskip

In a sense, dominated strategies play an analogous role to the usual dominated models in parametric statistical inference. The main advantage is that one can use the conditional density $f_n(\cdot\mid x)$ instead of the conditional measure $\sigma_n(x)$. A related advantage is that, if one fixes $\lambda$ and restricts to strategies dominated by $\lambda$, Problem (*) becomes simpler. However, even in applied data analysis, various familiar strategies are not dominated. In the framework of species sampling sequences, for instance, most strategies are not dominated. Therefore, in this paper, we focus on general strategies while the dominated ones are regarded as an important special case.

\medskip

\subsection{Content of this paper and further notation}\label{zq219uj}
This is a review paper on P.A. which also includes some (minor) new results. Our perspective is mainly on the probabilistic aspects of Bayesian predictive constructions. Moreover, we tacitly assume that the major target is to predict future observations (and not to make inference on other random elements, such as random parameters).

\medskip

Essentially, we aim to achieve three goals. First, we try to put P.A. in the right framework, to provide a unifying view, and to make clear a few misunderstandings. This has been done in the Introduction. Second, in Section \ref{d5tn9j} and Subsection \ref{z310jp}, we report some known results. Third, we provide some new strategies and we prove a few related results. The strategies, introduced by means of examples, deal with generalized P\'olya urns, random change points, covariates and stationary sequences. The results consist in determining the distribution of the data sequence $X$ under such strategies. To our knowledge, Examples \ref{wh77g5}, \ref{h8uj92w}, \ref{3w34r6gc1}, \ref{e41z9mht} and Theorems \ref{z34y9n1s}, \ref{g78j92xcd1q}, \ref{c5rt7n} are actually new, while Theorem \ref{7y7u8bf} makes precise a claim contained in \cite{FHW2021}. Moreover, as far as we know, Section \ref{x76t} is the first attempt to develop P.A. for stationary data. It provides a brief discussion of Problem (*) and introduces two large classes of stationary sequences.

\medskip

As already noted, even if $X$ could be potentially given any distribution, in most applications $X$ is required some conditions. There is obviously a number of such conditions. Among them, we decided to focus on exchangeability, stationarity and conditional identity in distribution. This choice seems reasonable to keep the paper focused, but of course it leaves out various interesting conditions, such as partial exchangeability. To write a paper of reasonable length, however, some choice was necessary.

\medskip

To defend our choice, we note that, in addition to be natural in various practical problems, exchangeability is the usual assumption in Bayesian prediction. Hence, taking exchangeability into account is more or less mandatory. Moreover, since $X$ is exchangeable if and only if it is stationary and c.i.d., the other two conditions can be motivated as the basic components of exchangeability. But there are also other reasons for dealing with them. Stationarity is in fact a routine assumption in the classical treatment of time series, and it is reasonable to consider it from the Bayesian point of view as well. Conditional identity in distribution, even if not that popular, seems to be quite suitable for P.A.; see Section \ref{xe5ty7}.

\medskip

The rest of the paper is organized in three sections, each concerned with a specific assumption on $X$, plus a final section of open problems. All the proofs are gathered in the Appendix.

\medskip

We close this Introduction with some further notations.

\medskip

As usual, $\delta_u$ is the unit mass at the point $u$. For each $x\in S^n$, where $n$ is a positive integer or $n=\infty$, we denote by $x_i$ the $i$-th coordinate of $x$. Moreover, we take $X$ to be the sequence of coordinate random variables on $S^\infty$, namely,
\begin{gather*}
X_i(x)=x_i\quad\quad\text{for all }i\ge 1\text{ and }x\in S^\infty.
\end{gather*}

From now on, we fix a strategy $\sigma$ and we assume
\begin{gather*}
X\overset{d}= P_\sigma.
\end{gather*}
We write $\nu$ instead of $\sigma_0$ (i.e., we let $\sigma_0=\nu$). Hence, $\nu$ is a probability measure on $\mathcal{B}$ to be regarded as the distribution of $X_1$ under the strategy $\sigma$. Finally, to avoid technicalities, $S$ is assumed to be a Borel subset of a Polish space and $\mathcal{B}$ the Borel $\sigma$-field on $S$.

\medskip

\section{Exchangeable data}\label{d5tn9j}

A permutation of $S^n$ is a map $\phi:S^n\rightarrow S^n$ of the form
\begin{gather*}
\phi(x)=(x_{j_1},\ldots,x_{j_n})\quad\quad\text{for all }x\in S^n
\end{gather*}
where $(j_1,\ldots,j_n)$ is a fixed permutation of $(1,\ldots,n)$. A sequence $Y=(Y_1,Y_2,\ldots)$ of random variables is {\em exchangeable} if
\begin{gather*}
\phi(Y_1,\ldots,Y_n)\overset{d}= (Y_1,\ldots,Y_n)
\end{gather*}
for all $n\ge 2$ and all permutations $\phi$ of $S^n$.

\medskip

As noted in Subsection \ref{xxz34}, if $X$ is required to be exchangeable, applying P.A. is usually hard. But there are a few exceptions and two of them are discussed in this section. We first recall that $X$ is a Dirichlet sequence (or a P\'olya sequence, see \cite{BMQ}) if
\begin{gather*}
\sigma_n(x)=\frac{c\,\nu+\sum_{i=1}^n\delta_{x_i}}{n+c}\quad\quad\text{for all }n\ge 0\text{ and }x\in S^n,
\end{gather*}
where $c>0$ is a constant, $\nu$ a probability measure on $\mathcal{B}$, and $\sigma_0(x)$ is meant as $\sigma_0(x)=\nu$. The role of Dirichlet sequences is actually huge in various frameworks, including Bayesian nonparametrics, population genetics, ecology, combinatorics and number theory; see e.g. [\citealp{FERG}, \citealp{GHOSVANDER}, \citealp{HHMW}, \citealp{PIT1996}, \citealp{PITYOR}, \citealp{PIT06}]. From our point of view, however, two facts are to be stressed. First, a Dirichlet sequence is exchangeable. Second, being defined through its predictive distributions, a Dirichlet sequence is a natural candidate for P.A.

\medskip

\subsection{Species sampling sequences}\label{g6l2xx5}
For $n\ge 1$ and $x=(x_1,\ldots,x_n)\in S^n$, denote by $k_n=k_n(x)$ the number of distinct values in the vector $x$ and by $x_1^*,\ldots,x_{k_n}^*$ such distinct values (in the order that they appear). Say that $X$ is a {\em species sampling sequence} if it is exchangeable, $\sigma_0=\nu$ is non-atomic, and
\begin{gather*}
\sigma_n(x)=\sum_{j=1}^{k_n}p_{j,n}(x)\,\delta_{x_j^*}+q_n(x)\,\nu
\\\text{for all }n\ge 1\text{ and }x\in S^n
\end{gather*}
where the $p_{j,n}$ are non-negative measurable functions on $S^n$ and $q_n=1-\sum_{j=1}^{k_n}p_{j,n}$. Under this strategy, quoting from \cite[p. 253]{HANPIT}, $X$ can be regarded as: ``the sequence of species of individuals in a process
of sequential random sampling from some hypothetical infinite population of individuals of various species.
The species of the first individual to be observed is assigned a random tag $X_1=X_1^*$ distributed according
to $\nu$. Given the tags $X_1,\ldots,X_n$ of the first $n$ individuals observed, it is supposed that the next individual
is one of the $j$-th species observed so far with probability $p_{j,n}$, and one of a new species with probability $q_n$".

\medskip

A nice consequence of the definition is that $p_{j,n}(x)$ depends on $x$ only through the vector $(N_{1,n},\ldots,N_{k_n,n})$, where
$$N_{j,n}=N_{j,n}(x)=\text{card}\bigl\{k:1\le k\le n,\,x_k=x_j^*\bigr\}$$
is the number of times that $x_j^*$ appears in the vector $x$; see [\citealp{HANPIT}, \citealp{PIT1996}].

\medskip

The most popular example of species sampling sequence is probably the {\em two-parameter Poisson-Dirichlet}, introduced by Pitman in \cite{PIT2PAR}, which corresponds to the weights
\begin{gather*}
p_{j,n}(x)=\frac{N_{j,n}-b}{n+c}\quad\text{and}\quad q_n(x)=\frac{b\,k_n+c}{n+c}
\end{gather*}
where $b$ and $c$ are constants such that: either (i) $0\le b< 1$ and $c>-b$ or (ii) $b<0$ and $c=-m\,b$ for some integer $m\ge 2$. In this model, if $L$ denotes the number of distinct values appearing in the sequence $X$, one obtains $L\overset{a.s.}=\infty$ under (i) and $L\overset{a.s.}=m$ under (ii). Note also that $X$ reduces to a Dirichlet sequence in the special case $b=0$.

\medskip

Another example, due to \cite{GNED10}, is
\begin{gather*}
p_{j,n}(x)=\frac{(N_{j,n}+1)(n-k_n+b)}{n^2+bn+c}
\\\text{and}\quad q_n(x)=\frac{k_n^2-bk_n+c}{n^2+bn+c}
\end{gather*}
where $b>0$ and $c$ is such that $k^2+bk+c>0$ for all integers $k>0$. This time, unlike the two-parameter Poisson-Dirichlet, $L$ is a finite but non-degenerate random variable.

\medskip

In general, to obtain a species sampling sequence, the forecaster needs to select $\nu$ and the weights $p_{j,n}$. While the choice of $\nu$ is free (apart from non-atomicity) the $p_{j,n}$ are subjected to the constraint that $X$ should be exchangeable. (Incidentally, the choice of $p_{j,n}$ is a good example of the difficulty of applying P.A. when $X$ is required to be exchangeable). The usual method to select $p_{j,n}$ involves {\em exchangeable random partitions}. Let $\mathbb{N}=\bigl\{1,2,\ldots\bigr\}$ and let $\Pi$ be a random partition of $\mathbb{N}$. For each $n\ge 1$, call $\Pi_n$ the restriction of $\Pi$ to $\{1,\ldots,n\}$, namely, the random partition of $\{1,\ldots,n\}$ whose elements are of the form $\{1,\ldots,n\}\cap A$ for some $A\in\Pi$. Say that $\Pi$ is exchangeable if
$$\varphi(\Pi_n)\overset{d}=\Pi_n$$
for all $n\ge 1$ and all permutations $\varphi$ of $(1,\ldots,n)$, where $\varphi(\Pi_n)$ denotes the random partition $\varphi(\Pi_n)=\bigl\{\varphi(B):B\in\Pi_n\bigr\}$. For instance, given any sequence $Y=(Y_1,Y_2,\ldots)$ of random variables, define $\Pi$ to be the random partition of $\mathbb{N}$ induced by the equivalence relation $i\sim j$ $\Leftrightarrow$ $Y_i=Y_j$. Then, $\Pi$ is exchangeable provided $Y$ is exchangeable. Now, the weights $p_{j,n}$ of a species sampling sequence correspond, in a canonical way, to the probability law of an exchangeable partition; see [\citealp{PIT2PAR}, \citealp{PIT1996}]. Hence, choosing the $p_{j,n}$ essentially amounts to choosing an exchangeable partition. We stop here since a detailed discussion of exchangeable partitions is bejond the scopes of this paper. The interested reader is referred to [\citealp{GNEDPIT}, \citealp{GNED10}, \citealp{QMT13}, \citealp{AIS08}, \citealp{PIT2PAR}, \citealp{PIT06}] and references therein.

\medskip

A last remark is that the definition of species sampling sequences can be generalized. In particular, non-atomicity of $\nu$ can be dropped (as in \cite{BL20} and \cite{CLNP}) and exchangeability can be replaced by some weaker condition (as in \cite{ACBLG2014} and \cite{BCL}).

\medskip

\subsection{Kernel based Dirichlet sequences}\label{zd5n88y}

In \cite{DIRGEN}, to generalize Dirichlet sequences while preserving their main properties, a class of strategies has been introduced. Among other things, such strategies make $X$ exchangeable.

\medskip

A {\em kernel} $\alpha$ on $(S,\mathcal{B})$ is a collection
$$\alpha=\bigl\{\alpha(\cdot\mid x):x\in S\bigr\}$$
such that $\alpha(\cdot\mid x)$ is a probability measure on $\mathcal{B}$, for each $x\in S$, and the map $x\mapsto\alpha(A\mid x)$ is measurable for each $A\in\mathcal{B}$. Sometimes, to make the notation easier, we will write $\alpha_x$ instead of $\alpha(\cdot\mid x)$. A straightforward example of kernel is $\alpha_x=\delta_x$ for each $x\in S$.

\medskip

Fix a probability measure $\nu$ on $\mathcal{B}$, a constant $c>0$, a kernel $\alpha$ on $(S,\mathcal{B})$, and define the strategy
\begin{gather}\label{v690k3e}
\sigma_n(x)=\frac{c\,\nu+\sum_{i=1}^n\alpha_{x_i}}{n+c}
\end{gather}
for all $n\ge 0$ and $x\in S^n$. Clearly, $X$ reduces to a Dirichlet sequence if $\alpha=\delta$. In this case, we also say that $X$ is a {\em classical} Dirichlet sequence.

\medskip

If $\alpha$ is an arbitrary kernel, $X$ may fail to be exchangeable. However, a useful sufficient condition for exchangeability is available. In fact, $X$ is exchangeable if $\alpha$ agrees with the conditional distribution for $\nu$ given some sub-$\sigma$-field $\mathcal{G}\subset\mathcal{B}$. For instance, if $\mathcal{G}=\mathcal{B}$, then $\alpha=\delta$ and $X$ is a classical Dirichlet sequence. At the opposite extreme, if $\mathcal{G}$ is the trivial $\sigma$-field, then $\alpha_x=\nu$ for all $x\in S$ and $X$ is i.i.d. with common distribution $\nu$. In general, for fixed $\nu$ and $c$, a strategy $\sigma$ which makes $X$ exchangeable can be associated with any sub-$\sigma$-field $\mathcal{G}\subset\mathcal{B}$. It suffices to take $\alpha$ as the conditional distribution for $\nu$ given $\mathcal{G}$.

\begin{ex}\label{g6s3}\textbf{(Countable partitions).}
Let $\mathcal{H}$ be a (non-random) countable partition of $S$ such that $H\in\mathcal{B}$ and $\nu(H)>0$ for all $H\in\mathcal{H}$. For $x\in S$, denote by $H_x$  the only $H\in\mathcal{H}$ such that $x\in H$. The conditional distribution for $\nu$ given the sub-$\sigma$-field generated by $\mathcal{H}$ is
\begin{gather*}
\alpha(\cdot\mid x)=\sum_{H\in\mathcal{H}}1_H(x)\,\nu\bigl(\cdot\mid H\bigr)=\nu\bigl(\cdot\mid H_x\bigr)\quad\text{for all }x\in S.
\end{gather*}
Hence, $X$ is exchangeable whenever
\begin{gather*}
\sigma_n(x)=\frac{c\,\nu+\sum_{i=1}^n\nu\bigl(\cdot\mid H_{x_i}\bigr)}{n+c}\quad\text{for all }n\ge 0\text{ and }x\in S^n.
\end{gather*}
Some remarks on the above strategy $\sigma$ are in order.
\begin{itemize}

\item $\sigma$ may be reasonable when the basic information provided by each observation $x_i$ is $H_{x_i}$, namely, the element of the partition $\mathcal{H}$ including $x_i$.

\item If $S$ is countable, each sub-$\sigma$-field $\mathcal{G}\subset\mathcal{B}$ is generated by a partition $\mathcal{H}$ of $S$. Hence, $\alpha$ is necessarily as above.

\item $\sigma_n(x)$ is absolutely continuous with respect to $\nu$ for all $n$ and $x$. This is a striking difference with classical Dirichlet sequences. To make an example, call $\sigma^*$ the strategy obtained by $\sigma$ replacing $\alpha$ with $\delta$. Under $\sigma^*$, $X$ is a classical Dirichlet sequence. Moreover, suppose $\nu$ is nonatomic and define the set $B(x)=\{x_1,\ldots,x_n\}$ for each $x=(x_1,\ldots,x_n)\in S^n$. Since $\nu$ is nonatomic and $B(x)$ is finite,
\begin{gather*}
P_\sigma\Bigl(X_{n+1}= X_i\text{ for some }i\le n\mid (X_1,\ldots,X_n)=x\Bigr)
\\=\sigma_n\bigl(x,\,B(x)\bigr)=0.
\end{gather*}
On the other hand, since $\delta_{x_i}(B(x))=1$ for each $i=1,\ldots,n$,
\begin{gather*}
P_{\sigma^*}\Bigl(X_{n+1}= X_i\text{ for some }i\le n\mid (X_1,\ldots,X_n)=x\Bigr)
\\=\sigma_n^*\bigl(x,B(x)\bigr)=n/(n+c).
\end{gather*}
As a consequence, one obtains
\begin{gather*}
P_\sigma\Bigl(\text{all the observations are distinct}\Big)=1,
\\P_{\sigma^*}\Bigl(\text{all the observations are distinct}\Big)=0.
\end{gather*}

\item $\sigma$ can be generalized replacing $\alpha$ with
\begin{gather*}
\beta(\cdot\mid x)=1_A(x)\,\delta_x\,+\,1_{A^c}(x)\,\nu\bigl(\cdot\mid A^c\cap H_x\bigr),
\end{gather*}
where $A\in\mathcal{B}$ is a suitable set. Note that $\beta$ reduces to $\alpha$ if $A=\emptyset$. Roughly speaking, $\beta$ is reasonable in those problems where there is a set $A$ such that $x_i$ is informative about the future observations only if $x_i\in A$. Otherwise, if $x_i\notin A$, the only relevant information provided by $x_i$ is $H_{x_i}$. As a trivial example, take $S=\mathbb{R}$ and
\begin{gather*}
\mathcal{H}=\bigl\{(-\infty,0),\,\{0\},\,(0,\infty)\bigr\},\quad A=[-u,u]
\end{gather*}
for some $u>0$. Then, $\beta$ is reasonable if $x_i$ is informative only if $\abs{x_i}\le u$. Otherwise, if $\abs{x_i}>u$, the only meaningful information provided by $x_i$ is its sign.
\end{itemize}
\end{ex}

\begin{ex}\label{gf5y7}\textbf{(P\'olya urns).} Some P\'olya urns are covered by Example \ref{g6s3}. It follows that, for such urns, the sequence $X$ of observed colors is exchangeable. To our knowledge, this fact was previously unknown.

\medskip

As an example, consider sequential draws from an urn and denote by $X_n$ the color of the ball extracted at time $n\ge 1$. At time $n=0$, the urn contains $m_j$ balls of color $j$ where $j\in\{1,\ldots,k\}$. Define
\begin{gather*}
S=\{1,\ldots,k\},\quad m=\sum_{j=1}^km_j\quad\text{and}\quad\nu\{j\}=\frac{m_j}{m}
\end{gather*}
for each $j\in S$. The sampling scheme is as follows. Fix a partition $\mathcal{H}$ of $S$ and define
\begin{gather*}
m_j^*=m\,\nu\bigl(\{j\}\mid H_j\bigr)=\frac{m\,m_j}{\sum_{i\in H_j}m_i}.
\end{gather*}
For each $n\ge 1$, one obtains $X_n\in H$ for some unique $H\in\mathcal{H}$. In this case (i.e., if $X_n\in H$) the extracted ball is replaced together with $m_j^*$ more balls of color $j$ for each $j\in H$. In other terms, if the observed color belongs to $H$, {\em each} color in $H$ is reinforced (and not only the observed color). In particular, after each draw, $m$ new balls are added to the urn. Hence, denoting by $\sigma$ the strategy of Example \ref{g6s3} with $c=1$, one obtains
\begin{gather*}
P\bigl(X_{n+1}=j\mid (X_1,\ldots,X_n)=x\bigr)
\\=\frac{m_j+\sum_{i=1}^n1_{H_j}(x_i)\,m_j^*}{m+m\,n}
\\=\frac{\nu\{j\}+\sum_{i=1}^n1_{H_j}(x_i)\,\nu\bigl(\{j\}\mid H_j\bigr)}{1+n}
\\=\frac{c\,\nu\{j\}+\sum_{i=1}^n\nu\bigl(\{j\}\mid H_{x_i}\bigr)}{c+n}=\sigma_n(x)\{j\}.
\end{gather*}
\end{ex}

\medskip

If $\sigma$ is the strategy \eqref{v690k3e}, in addition to exchangeability, $X$ satisfies various other properties of classical Dirichlet sequences. We refer to \cite{DIRGEN} for details. Here, we just note that the prior $\pi$ and the posterior $\pi_n$ can be explicitly determined. In particular, up to replacing $\delta$ with $\alpha$, the Sethuraman's representation of $\pi$ (see \cite{SET}) is still true. Precisely, $\pi$ is the probability distribution of a random probability measure $\mu$ of the form
\begin{gather*}
\mu(\cdot)=\sum_jV_j\,\alpha(\cdot\mid Z_j)
\end{gather*}
where:

\medskip

\begin{itemize}

\item $(Z_j)$ and $(V_j)$ are independent sequences of random variables;

\item $(Z_j)$ is i.i.d. with common distribution $\nu$;

\item $V_j=U_j\,\prod_{i=1}^{j-1}(1-U_i)$ for all $j\ge 1$, where $(U_i)$ is i.i.d. with common distribution beta$(1,c)$. Namely, $(V_j)$ has the {\em stick breaking distribution} with parameter $c$.

\end{itemize}

\section{Conditionally identically distributed data}\label{xe5ty7}

A sequence $Y=(Y_1,Y_2,\ldots)$ of random variables is {\em conditionally identically distributed} (c.i.d.) if $Y_2\overset{d}= Y_1$ and
\begin{equation*}
P\bigl(Y_k\in\cdot\mid Y_1,\ldots,Y_n\bigr)=P\bigl(Y_{n+1}\in\cdot\mid Y_1,\ldots,Y_n\bigr)\quad\text{a.s.}
\end{equation*}
for all $k>n\ge 1$. A c.i.d. sequence $Y$ is identically distributed. It is also asymptotically exchangeable in the sense that, as $n\rightarrow\infty$, the probability distribution of the shifted sequence $(Y_n,Y_{n+1},\ldots)$ converges weakly to an exchangeable law. Moreover, as already stressed, $Y$ is exchangeable if and only if it is stationary and c.i.d.

\medskip

C.i.d. sequences have been introduced in [\citealp{BPR2004}, \citealp{K}] and then investigated or applied in various papers; see e.g. [\citealp{ACBLG2014}, \citealp{BCL}, \citealp{BPR2012}, \citealp{BPR2013}, \citealp{BDPR2021}, \citealp{BPCID}, \citealp{CZGV}, \citealp{CSZ2023}, \citealp{FHW2021}, \citealp{FL2022}, \citealp{FPS2018}].

\medskip

There are reasons for taking c.i.d. data into account in Bayesian prediction. In fact, in a sense, c.i.d. sequences have been introduced having prediction in mind. If $X$ is c.i.d., at each time $n$, the future observations $(X_k:k>n)$ are identically distributed given the past, and this is reasonable in several prediction problems. Examples arise in clinical trials, generalized P\'olya urns, species sampling models, survival analysis and disease surveillance; see [\citealp{ACBLG2014}, \citealp{BCL}, \citealp{BPR2004}, \citealp{BDPR2021}, \citealp{BPCID}, \citealp{CZGV}, \citealp{CSZ2023}, \citealp{FHW2021}, \citealp{FL2022}, \citealp{FP2020}]. A further reason for assuming $X$ c.i.d. is that the asymptotics is very close to that of exchangeable sequences. As a consequence, a meaningful part of the usual Bayesian machinery can be developed under the sole assumption that $X$ is c.i.d.; see \cite{FHW2021}. Finally, the strategies which make $X$ c.i.d. can be easily characterized; see Theorem \ref{l77be4c} in the Appendix. Hence, unlike the exchangeable case, P.A. can be easily implemented for c.i.d. data. A number of interesting strategies, which cannot be used if $X$ is required to be exchangeable, become available if $X$ is only asked to be c.i.d.; see e.g. [\citealp{BDPR2021}, \citealp{BPCID}].

\medskip

As a concrete example, fix a constant $q\in (0,1)$ and define
\begin{gather}\label{gy3cd1z0}
\sigma_n(x)=q^n\nu+(1-q)\sum_{i=1}^nq^{n-i}\delta_{x_i}
\end{gather}
for all $n\ge 0$ and $x\in S^n$. Using $\sigma$ to make predictions corresponds to exponential smoothing. It may be reasonable when the forecaster has only vague opinions on the dependence structure of the data, and yet she feels that the weight of the $i$-th observation $x_i$ should be a decreasing function of $n-i$. In this case, $X$ is not exchangeable, since $\sigma_n(x)$ is not invariant under permutation of $x$, but it can be easily seen to be c.i.d.; see \cite[Ex. 7]{BDPR2021}.

\medskip

In this section, following [\citealp{BDPR2021}, \citealp{BPCID}], P.A. is applied to c.i.d. data. We first report some known strategies (Subsection \ref{z310jp}) and then we introduce two new strategies which make $X$ c.i.d. (Subsection \ref{a21j9b}).

\medskip

\subsection{Fast recursive update of predictive distributions}\label{z310jp}

A possible condition for a strategy $\sigma$ is
\begin{gather}\label{bg7y78u}
\sigma_{n+1}(x,y)\text{ is a function of }\sigma_n(x)\text{ and }y
\end{gather}
for all $n\ge 0$, $x\in S^n$ and $y\in S$, where $y$ denotes the $(n+1)$-th observation and
\begin{gather*}
(x,y)=(x_1,\ldots,x_n,y).
\end{gather*}
Under \eqref{bg7y78u}, the predictive $\sigma_{n+1}(x,y)$ is just a recursive update of the previous predictive $\sigma_n(x)$ and the last observation $y$. Recursive properties of this type are useful in applications. They have a long history (see e.g. [\citealp{NZ1999}, \citealp{N2002}, \citealp{SM1978}]) and have been recently investigated in \cite{HMW}.

\medskip

For each $n\ge 0$, let $q_n:S^n\rightarrow [0,1]$ be a measurable function (with $q_0$ constant) and $\alpha_n$ a kernel on $(S,\mathcal{B})$. Define a strategy $\sigma$ through the recursive equations
\begin{gather}\label{rb7z22k}
\sigma_0=\nu\quad\quad\text{and}
\\\sigma_{n+1}(x,y)=q_n(x)\,\sigma_n(x)+(1-q_n(x))\,\alpha_n(\cdot\mid y)\notag
\end{gather}
for all $n\ge 0$, $x\in S^n$ and $y\in S$. Since $\sigma_{n+1}(x,y)$ is a convex combination of the previous predictive $\sigma_n(x)$ and the kernel $\alpha_n(\cdot\mid y)$, which depends only on $y$, the strategy $\sigma$ satisfies condition \eqref{bg7y78u}. The obvious interpretation is that, at time $n+1$, after observing $(x,y)$, the next observation is drawn from $\sigma_n(x)$ with probability $q_n(x)$ and from $\alpha_n(\cdot\mid y)$ with probability $1-q_n(x)$.

\medskip

An example of strategy satisfying equation \eqref{rb7z22k} is Newton's algorithm [\citealp{NZ1999}, \citealp{N2002}]. More precisely, Newton's algorithm aims to estimate the latent distribution in a mixture model rather than to make predictions. However, if reinterpreted as a predictive rule, Newton's algorithm corresponds to a strategy $\sigma$ and such a $\sigma$ meets equation \eqref{rb7z22k} for a suitable choice of $q_n$ and $\alpha_n$; see e.g. \cite[p. 1095]{FP2020}. Moreover, as shown in \cite{FP2020}, $\sigma$ makes $X$ c.i.d.

\medskip

The strategies satisfying equation \eqref{rb7z22k} are investigated in \cite{BPCID}. Under such strategies, $X$ is usually not exchangeable but it is c.i.d. under some conditions on the kernels $\alpha_n$. Precisely, $X$ is c.i.d. if  $\alpha_n$ is the conditional distribution for $\nu$ given $\mathcal{G}_n$ for each $n\ge 0$, where
\begin{gather*}
\mathcal{G}_0\subset\mathcal{G}_1\subset\mathcal{G}_2\subset\ldots\subset\mathcal{B}
\end{gather*}
is any filtration (i.e., any increasing sequence of sub-$\sigma$-fields of $\mathcal{B}$). This condition is trivially true if $\alpha_n(\cdot\mid y)=\delta_y$ for all $y\in S$ (just take $\mathcal{G}_n=\mathcal{B}$ for all $n\ge 0$).

\medskip

\begin{ex}\label{v6yh8}\textbf{(Finer countable partitions).}
For each $n\ge 0$, let $\mathcal{H}_n$ be a countable partition of $S$ such that $H\in\mathcal{B}$ and $\nu(H)>0$ for all $H\in\mathcal{H}_n$. Suppose that $\mathcal{H}_{n+1}$ is finer than $\mathcal{H}_n$ for all $n\ge 0$. Define $\sigma$ through equation \eqref{rb7z22k} with
\begin{gather*}
\alpha_n(\cdot\mid y)=\sum_{H\in\mathcal{H}_n}1_H(y)\,\nu(\cdot\mid H)=\nu\bigl(\cdot\mid H_y^n\bigr)
\end{gather*}
where $H_y^n$ denotes the only $H\in\mathcal{H}_n$ such that $y\in H$. The kernel $\alpha_n$ is the conditional distribution for $\nu$ given $\mathcal{G}_n$, where $\mathcal{G}_n$ is the $\sigma$-field generated by $\mathcal{H}_n$. Since $\mathcal{H}_{n+1}$ is finer than $\mathcal{H}_n$, one obtains $\mathcal{G}_n\subset\mathcal{G}_{n+1}$. Hence, $X$ is c.i.d. Note also that the $\mathcal{H}_n$ could be chosen such that
\begin{gather*}
\{y\}=\bigcap_n H_y^n\quad\quad\text{for all }y\in S.
\end{gather*}
In this case, as $n\rightarrow\infty$, the partitions $\mathcal{H}_n$ shrink to the partition of $S$ in the singletons.

For instance, in Example \ref{g6s3}, suppose the forecaster wants to replace the fixed partition $\mathcal{H}$ with a sequence $\mathcal{H}_n$ of finer partitions. This is possible at the price of having $X$ c.i.d. instead of exchangeable. In fact, with $q_n=\frac{n+c}{n+1+c}$, one obtains
\begin{gather*}
\sigma_n(x)=\frac{c\,\nu+\sum_{i=1}^n\alpha_{i-1}(\cdot\mid x_i)}{n+c}
\\=\frac{c\,\nu+\sum_{i=1}^n\nu\bigl(\cdot\mid H_{x_i}^{i-1}\bigr)}{n+c}.
\end{gather*}
Similarly, to decrease the impact of the observed data while preserving the c.i.d. condition, the strategy \eqref{gy3cd1z0} could be modified as
\begin{gather*}
\sigma_n(x)=q^n\nu+(1-q)\sum_{i=1}^nq^{n-i}\nu\bigl(\cdot\mid H_{x_i}^{i-1}\bigr).
\end{gather*}
\end{ex}

\medskip

We next turn to a strategy introduced in \cite{HMW}. Once again, under this strategy, the data are c.i.d. but not necessarily exchangeable.

\medskip

\begin{ex}\label{f6y7ua1m}\textbf{(Hahn, Martin and Walker; Copulas).}
In this example, $S=\mathbb{R}$ and ``density function'' means ``density function with respect to Lebesgue measure''. A bivariate {\em copula} is a distribution function on $\mathbb{R}^2$ whose marginals are uniform on $(0,1)$. The density function of a bivariate copula, provided it exists, is said to be a {\em copula density}.

\medskip

In \cite{HMW}, in order to realize condition \eqref{bg7y78u}, the following updating rule is introduced. Fix a density $f_0$ and a sequence $c_1,c_2,\ldots$ of bivariate copula densities. For the sake of simplicity, we assume $f_0>0$ and $c_n>0$ for all $n\ge 1$. For $n=0$, define $\sigma_0(dz)=f_0(z)\,dz$ and call $F_0$ the distribution function corresponding to $\sigma_0$. Then, for each $y\in\mathbb{R}$, define
\begin{gather*}
\sigma_1(y,\,dz)=f_1(z\mid y)\,dz\quad\quad\text{where}
\\f_1(z\mid y)=c_1\bigl\{F_0(z),\,F_0(y)\bigr\}\,f_0(z).
\end{gather*}
In general, for each $n\ge 0$ and $x\in\mathbb{R}^n$, suppose $\sigma_n(x)$ has been defined and denote by $f_n(\cdot\mid x)$ and $F_n(\cdot\mid x)$ the density and the distribution function of $\sigma_n(x)$. Then, for all $y\in\mathbb{R}$, one can define
\begin{gather}\label{g6n83k0u}
\sigma_{n+1}(x,y,\,dz)=f_{n+1}(z\mid x,y)\,dz\quad\text{where}
\\f_{n+1}(z\mid x,y)=c_{n+1}\bigl\{F_n(z\mid x),\,F_n(y\mid x)\bigr\}\,f_n(z\mid x).\notag
\end{gather}
Equation \eqref{g6n83k0u} defines a strategy $\sigma$ dominated by the Lebesgue measure.

\medskip

In \cite{HMW} (but not here) the $c_n$ are also required to be symmetric. Furthermore, in \cite{HMW}, equation \eqref{g6n83k0u} is not necessarily viewed as a method for obtaining a strategy but is {\em deduced} as a consequence of exchangeability. From our point of view, instead, equation \eqref{g6n83k0u} defines a strategy $\sigma$ which we call HMW's strategy.

\medskip

Under HMW's strategy, $X$ is not necessarily exchangeable, even if the $c_n$ are symmetric and $c_n\rightarrow 1$ (in some sense) as $n\rightarrow\infty$. To see this, recall that $X$ is i.i.d. if and only if it is exchangeable and $X_1$ is independent of $X_2$. In turn, $X_1$ is independent of $X_2$ if $c_1$ is the independence copula density (i.e., $c_1(u,v)=1$ for all $(u,v)\in [0,1]^2$). Therefore, $X$ fails to be exchangeable whenever $c_1$ is the independence copula density and $c_2\neq c_1$. However, as noted in \cite{FHW2021}, $X$ turns out to be c.i.d.

\medskip

\begin{thm}\label{7y7u8bf} If $\sigma$ is HMW's strategy, then $X$ is c.i.d. \end{thm}

\medskip

A proof of Theorem \ref{7y7u8bf} is provided in the Appendix. We note that, for Theorem \ref{7y7u8bf} to hold, the positivity assumption on $f_0$ and $c_n$ may be dropped and the $c_n$ can be taken to be conditional copula densities; see Remark \ref{g73sd}.
\end{ex}

\medskip

\subsection{Further examples}\label{a21j9b}

In the next example, the data are exchangeable until a stopping time $T$ and then go on so as to form a c.i.d. sequence. The time $T$ should be regarded as the first time when something meaningful happens, possibly something modifying the nature of the observed phenomenon. Even if apparently involved, the example could find some applications. For instance, to model censored survival times, with $T-1$ the first time when a given number of survival times is observed.

\medskip

\begin{ex}\textbf{(Change points).}\label{wh77g5} A predictable stopping time is a function $T$ on $S^\infty$, with values in $\{2,3,\ldots,\infty\}$, satisfying
\begin{gather}\label{n6y9m1}
\bigl\{T={n+1}\bigr\}=\bigl\{(X_1,\ldots,X_n)\in A_n\bigr\}
\end{gather}
for some set $A_n\in\mathcal{B}^n$. Basically, condition \eqref{n6y9m1} means that the event $\{T=n+1\}$ depends only on $(X_1,\ldots,X_n)$. Similarly, $\{T\le n+1\}=\bigcup_{j=2}^{n+1}\{T=j\}$ depends only on $(X_1,\ldots,X_n)$. Therefore, for all $x\in S^n$ and $y\in S$, the indicators of $\{T\le n+1\}$ and $\{T>n+1\}$ depend on $x$ but not on $y$.

\medskip

Fix a predictable stopping time $T$ and a strategy $\beta=(\beta_0,\beta_1,\ldots)$ which makes $X$ exchangeable. Moreover, as in Subsection \ref{z310jp}, fix the measurable functions $q_n:S^n\rightarrow [0,1]$. Then, define $\sigma_0=\beta_0$, $\sigma_1=\beta_1$, and
\begin{gather*}
\sigma_{n+1}(x,y)=1_{\{T>n+1\}}(x)\,\beta_{n+1}(x,y)\,+
\\+\,1_{\{T\le n+1\}}(x)\,\Bigl\{q_n(x)\,\sigma_n(x)\,+\,(1-q_n(x))\,\delta_y\Bigr\}
\end{gather*}
for all $n\ge 1$, $x\in S^n$ and $y\in S$. In the Appendix, it is shown that:

\medskip

\begin{thm}\label{z34y9n1s} The above strategy $\sigma$ makes $X$ c.i.d. Moreover, if
\begin{gather*}
A_n\text{ is invariant under permutations of }S^n\text{ for all }n\ge 1,
\end{gather*}
where $A_n$ is the set involved in condition \eqref{n6y9m1}, then $(X_1,\ldots,X_n)$ is exchangeable conditionally on $T>n$. Precisely,
\begin{gather*}
P_\sigma\Bigl(\phi(X_1,\ldots,X_n)\in\cdot\mid T>n\Bigr)
\\=P_\sigma\Bigl((X_1,\ldots,X_n)\in\cdot\mid T>n\Bigr)
\end{gather*}
for all $n$ such that $P_\sigma(T>n)>0$ and all permutations $\phi$ of $S^n$.
\end{thm}

\medskip

Theorem \ref{z34y9n1s} is still valid if $\sigma$ is defined differently at the times subsequent to $T$. For instance, given a countable partition $\mathcal{H}$ of $S$, the conclusions of Theorem \ref{z34y9n1s} are true even if
\begin{gather*}
\sigma_{n+1}(x,y)=q_n(x)\,\sigma_n(x)+(1-q_n(x))\,\sigma_n(x,\,\cdot\mid H_y)
\end{gather*}
for all $x\in S^n$ and $y\in S$ such that $T\le n+1$ and $\sigma_n(x,\,H_y)>0$. Here, $\sigma_n(x,\,\cdot\mid H_y)$ denotes the probability measure
$$\sigma_n(x,\,A\mid H_y)=\frac{\sigma_n(x,\,A\cap H_y)}{\sigma_n(x,\,H_y)}\quad\quad\text{for all }A\in\mathcal{B}.$$

\medskip

Censored survival times are a possible application of $\sigma$. Suppose that $S=\{0,1\}\times (0,\infty)$ and the $i$-th observation is a pair $x_i=(j_i,t_i)$ where $t_i$ is the survival time of item $i$, or the time when item $i$ leaves the trial, according to whether $j_i=1$ or $j_i=0$. In this framework, $T-1$ could be the first time when a fixed number $k$ of survival times is observed, namely,
\begin{gather*}
T=1+\inf\bigl\{n:\sum_{i=1}^nj_i=k\bigr\}
\end{gather*}
with the usual convention $\inf\emptyset=\infty$. Finally, the strategy $\beta$ could be as in Subsection \ref{zd5n88y}. In fact, classical Dirichlet sequences are a quite popular model to describe censored survival times but have the drawback of ties. This drawback may be overcome if $\beta$ is of the form
$$\beta_n(x)=\frac{c\,\nu+\sum_{i=1}^n\alpha_{x_i}}{n+c},$$
where the kernel $\alpha$ satisfies the conditions of Subsection \ref{zd5n88y} and $\nu$ and $\alpha_x$ are nonatomic for all $x\in S$.
\end{ex}

\medskip

So far, the $n$-th predictive distribution has been meant as the conditional distribution of $X_{n+1}$ given $(X_1,\ldots,X_n)$. But the information available at time $n$ is often strictly larger than $(X_1,\ldots,X_n)$. To model this situation, we suppose to observe the sequence $$Y=(X_1,Z_1,X_2,Z_2,\ldots)$$
where $Z=(Z_1,Z_2,\ldots)$ is any sequence of random variables. The $Z_n$ can be regarded as covariates. At each time $n$, the forecaster aims to predict $X_{n+1}$ based on $(X_1,Z_1,\ldots,X_n,Z_n)$. She is not interested in $Z_{n+1}$ as such, but $Z_1,\ldots,Z_n$ can not be neglected since they are informative on $X_{n+1}$. Moreover, she wants $X$ to be c.i.d. and $Z$ unconstrained as much as possible. One solution could be a strategy which makes $Y$ c.i.d. However, if $Y$ is c.i.d., both $X$ and $Z$ are marginally c.i.d., and having $Z$ c.i.d. may be unwelcome. In the next example, $X$ is c.i.d. but $Z$ is not. In addition, $X$ satisfies a condition stronger than the c.i.d. one, that is, $X_2\overset{d}= X_1$ and
\begin{gather}\label{g7h8m}
P\bigl(X_k\in\cdot\mid X_1,Z_1,\ldots,X_n,Z_n\bigr)\\=P\bigl(X_{n+1}\in\cdot\mid X_1,Z_1,\ldots,X_n,Z_n\bigr)\notag
\end{gather}
a.s. for all $k>n\ge 1$; see \cite{BPR2004}.

\begin{ex}\textbf{(Covariates).}\label{h8uj92w}
Let $S=\mathbb{R}^2$ and
\begin{gather*}
0=b_0<b_1<b_2<\ldots,\quad\sup_nb_n\le 1,
\end{gather*}
a bounded strictly increasing sequence of real numbers. Take $\sigma_0$ as the probability distribution of $(U+V,V)$ where
\begin{gather*}
U\text{ independent of }V,\quad U\overset{d}=\mathcal{N}(0,b_1),\quad V\overset{d}=\mathcal{N}(0,1-b_1).
\end{gather*}
Similarly, for each $n\ge 1$ and
\begin{gather*}
y=(y_1,\ldots,y_n)=(x_1,z_1,\ldots,x_n,z_n),
\end{gather*}
take $\sigma_n(y)$ as the probability distribution of $(U_n(y)+V_n(y),\,V_n(y))$ where
\begin{gather*}
U_n(y)\text{ independent of }V_n(y),
\\U_n(y)\overset{d}=\mathcal{N}\bigl(x_n-z_n,\,b_{n+1}-b_n\bigr),
\\V_n(y)\overset{d}=\mathcal{N}(0,1-b_{n+1}).
\end{gather*}
Then, $Z$ is not c.i.d. while $X$ satisfies condition \eqref{g7h8m}. Furthermore, arguing as in \cite[Sect. 4]{BPCID}, the normal distribution could be replaced by any symmetric stable law.

\medskip

To see that $Z$ is not c.i.d., just note that $Z$ fails to be identically distributed. To prove condition \eqref{g7h8m}, take a collection $\bigl\{T_n, W_n:n\ge 1\bigr\}$ of independent standard normal random variables and define the sequence
$$Y^*=(X_1^*,Z_1^*,X_2^*,Z_2^*,\ldots),$$
where $Z_n^*=\sqrt{1-b_n}\,W_n$ and
$$X_n^*=\sum_{j=1}^n\sqrt{b_j-b_{j-1}}\,T_j+Z_n^*.$$
It is not hard to verify that $Y^*\overset{d}= Y$. Hence, it suffices to prove \eqref{g7h8m} with $Y^*$ in the place of $Y$, and this can be done as in \cite[Ex. 1.2]{BPR2004}. We omit the explicit calculations.
\end{ex}

\section{Stationary data}\label{x76t}

A sequence $Y=(Y_1,Y_2,\ldots)$ of random variables is {\em stationary} if
\begin{gather*}
(Y_2,\ldots,Y_{n+1})\overset{d}= (Y_1,\ldots,Y_n)\quad\quad\text{for all }n\ge 1.
\end{gather*}

\medskip

In the non-Bayesian approaches to prediction, stationarity is a classical assumption. In a Bayesian framework, instead, stationarity seems to be less popular. In particular, to our knowledge, there is no systematic treatment of P.A. for stationary data. This section aims to fill this gap and begins an investigation of P.A. when $X$ is required to be stationary. It is just a preliminary step and much more work is to be done.

\medskip

After some general remarks on Problem (*), two large classes of stationary sequences will be introduced. Incidentally, these two classes may look unusual for a Bayesian forecaster. We don't know whether this is true, but we recall that P.A. is consistent with any probability distribution for $X$. Hence, in a Bayesian framework, using data coming from such classes is certainly admissible.

\medskip

If $X$ is required to be stationary, for P.A. to apply, the strategies which make $X$ stationary should be characterized. Hence, one comes across Problem (*) with $\mathcal{C}$ the class of stationary probability measures on $(S^\infty,\mathcal{B}^\infty)$. This version of Problem (*) is quite hard and we are not aware of any general solution; see e.g. [\citealp{BLMC}, \citealp{MOWE}] and references therein. Fortunately, however, Problem (*) is simple (or even trivial) in a few special cases. As an example, a strategy $\sigma$ makes $X$ a stationary (first order) Markov chain if and only if
\begin{gather*}
\int\sigma_1(x,\,\cdot)\,\sigma_0(dx)=\sigma_0(\cdot)\quad\text{and}\quad\sigma_n(x)=\sigma_1(x_n)
\end{gather*}
for all $n\ge 1$ and $P_\sigma$-almost all $x\in S^n$. Even if obvious, this fact has a useful practical consequence. If the data are required to be stationary and Markov, in order to make Bayesian predictions, applying P.A. is straightforward.

\medskip

Another remark is that, unlike the exchangeable case, a finite dimensional stationary random vector can be always extended to an (infinite) stationary sequence. To formalize this fact, we first recall that the probability distribution of the random vector $(X_1,\ldots,X_n)$ is completely determined by $\sigma_0,\sigma_1,\ldots,\sigma_{n-1}$.

\begin{lem}\label{f6yn9ik3e}
Fix $n\ge 1$, select $\sigma_0,\sigma_1,\ldots,\sigma_{n-1}$ and define
\begin{gather*}
\sigma_j(u,x)=\sigma_{n-1}(x)
\end{gather*}
for all $j>n-1$, $u\in S^{j-n+1}$ and $x\in S^{n-1}$. Then, $X$ is stationary provided $(X_2,\ldots,X_n)\overset{d}= (X_1,\ldots,X_{n-1})$.
\end{lem}

Lemma \ref{f6yn9ik3e} is probably well known, but again we do not know of any explicit reference. Anyway, the proof is straightforward. It suffices to note that, under the strategy of Lemma \ref{f6yn9ik3e}, $X_{j+1}$ is conditionally independent of $(X_1,\ldots,X_{j-n+1})$ given $(X_{j-n+2},\ldots,X_j)$.

\medskip

A last remark is that Problem (*) admits an obvious solution for dominated strategies. In this case, incidentally, Problem (*) can be easily solved even for exchangeable data.

\begin{thm}\label{g78j92xcd1q}
Let $\lambda$ be a $\sigma$-finite measure on $(S,\mathcal{B})$ and $\sigma$ a strategy dominated by $\lambda$, say
\begin{gather*}
\sigma_0(dy)=f_0(y)\,\lambda(dy)\quad\text{and}\quad\sigma_n(x,\,dy)=f_n(y\mid x)\,\lambda(dy)
\end{gather*}
for all $n\ge 1$ and $x\in S^n$. Define
\begin{gather*}
g_n(x)=f_0(x_1)\,f_1(x_2\mid x_1)\ldots f_{n-1}(x_n\mid x_1,\ldots,x_{n-1})
\end{gather*}
for all $n\ge 1$ and $x\in S^n$. Then,

\begin{itemize}

\item $P_\sigma$ is stationary if and only if
\begin{gather*}
g_n(x)=\int g_{n+1}(u,x)\,\lambda(du)
\end{gather*}
for all $n\ge 1$ and $P_\sigma$-almost all $x\in S^n$.

\item $P_\sigma$ is exchangeable if and only if
\begin{gather*}
g_n(\phi(x))=g_n(x)
\end{gather*}
for all $n\ge 2$, all permutations $\phi$ of $S^n$ and $P_\sigma$-almost all $x\in S^n$.

\end{itemize}

\end{thm}

\medskip

The proof of Theorem \ref{g78j92xcd1q} is given in the Appendix.

\medskip

We finally give two examples. In both, $X$ is a stationary Markov sequence, possibly of order greater than 1.

\medskip

\begin{ex}\label{3w34r6gc1}\textbf{(Generalized autoregressive sequences).}
Let $S=\mathbb{R}$. Fix a probability measure $\mu$ on $\mathcal{B}$ and a measurable function $f:\mathbb{R}\rightarrow\mathbb{R}$. Define
\begin{gather*}
\sigma_1(x,A)=P\bigl(f(x)+U\in A)\quad\quad\text{for all }x\in\mathbb{R} \text{ and }A\in\mathcal{B},
\end{gather*}
where $U$ is a real random variable such that $U\overset{d}=\mu$. {\em Suppose} now that
\begin{gather}\label{g67y84f}
\int\sigma_1(x,A)\,\nu(dx)=\nu(A),\quad\quad A\in\mathcal{B},
\end{gather}
for some  probability measure $\nu$ on $\mathcal{B}$. Then, $X$ is a stationary Markov chain provided
\begin{gather*}
\sigma_0=\nu\quad\text{and}\quad\sigma_n(x)=\sigma_1(x_n)\quad\text{for all }n\ge 2\text{ and }x\in\mathbb{R}^n.
\end{gather*}
Note that $Y\overset{d}= P_\sigma$ for any sequence $Y=(Y_1,Y_2,\ldots)$ such that
\begin{gather*}
Y_1\overset{d}=\nu\quad\text{and}\quad Y_n=f(Y_{n-1})+U_n\text{ for }n\ge 2,
\end{gather*}
where $(U_n:n\ge 2)$ is i.i.d., independent of $Y_1$, and $U_2\overset{d}=\mu$. Thus, $\mu$ can be regarded as the distribution of the ``errors" $U_n$ and $\nu$ as the marginal distribution of the observations $Y_n$. For instance, the usual Gaussian (first order) autoregressive processes correspond to $f(x)=c\,x$, $\mu=\mathcal{N}(0,b)$ and $\nu=\mathcal{N}(0,\,b/(1-c^2))$, where $c\in (-1,1)$ and $b>0$ are constants.

\medskip

To make the above argument concrete, the following problem is to be solved: {\em For fixed $f$ and $\mu$, give conditions for the existence of $\nu$ satisfying equation \eqref{g67y84f}. More importantly, give an explicit formula for $\nu$ provided it exists}. We next focus on this problem in the (meaningful) special case where $\mu$ is a symmetric stable law.

\medskip

Let $\gamma\in (0,2]$ be a constant and $Z$ a real random variable with characteristic function
\begin{gather*}
E\bigl\{\exp(i\,t\,Z)\bigr\}=\exp\Bigl(-\frac{\abs{t}^\gamma}{2}\Bigr)\quad\quad\text{for all }t\in\mathbb{R}.
\end{gather*}
(The exponent $\gamma$ is usually denoted by $\alpha$, but this notation cannot be adopted in this paper since $\alpha$ denotes a kernel). For $a\in\mathbb{R}$ and $b>0$, denote by $\mathcal{S}(a,b)$ the probability distribution of $a+b^{1/\gamma}Z$, namely
\begin{gather*}
\mathcal{S}(a,b;\,A)=P\bigl(a+b^{1/\gamma}Z\in A)\quad\quad\text{for all }A\in\mathcal{B}.
\end{gather*}
The probability measure $\mathcal{S}(a,b)$ is said to be a symmetric stable law with exponent $\gamma$. Note that $\mathcal{S}(a,b)=\mathcal{N}(a,b)$ if $\gamma=2$ and $\mathcal{S}(a,b)=\mathcal{C}(a,b)$ if $\gamma=1$, where $\mathcal{C}(a,b)$ is the Cauchy distribution with density $f(x)=\frac{2\,b}{\pi}\,\frac{1}{b^2+4\,(x-a)^2}$ (the standard Cauchy distribution corresponds to $a=0$ and $b=2$).

\begin{thm}\label{c5rt7n}
Let $c\in (-1,1)$ be a constant. If $\mu=\mathcal{S}(a,b)$ and $f(x)=-a+c\,x$, then equation \eqref{g67y84f} is satisfied by
\begin{gather*}
\nu=\mathcal{S}\left(0,\,\frac{b}{1-\abs{c}^\gamma}\right).
\end{gather*}
\end{thm}

By Theorem \ref{c5rt7n}, which is proved in the Appendix, one obtains (first order) stationary autoregressive processes with any symmetric stable marginal distribution.
\end{ex}

\medskip

\begin{ex}\label{e41z9mht}\textbf{(Markov sequences of arbitrary order).}
Let $\lambda$ be a $\sigma$-finite measure on $(S,\mathcal{B})$. Fix $n\ge 2$ and a measurable function $h$ on $S^n$ such that $h>0$ and $\int h\,d\lambda^n=1$. Given $h$, define a further function $g$ via cyclic permutations of $h$, namely
\begin{gather*}
g(x)=\frac{1}{n}\,\bigl\{h(x_1,\ldots,x_n)+h(x_2,\ldots,x_n,x_1)+
\\+\ldots+h(x_n,x_1,\ldots,x_{n-1})\bigr\}
\end{gather*}
for all $x\in S^n$. Such a $g$ is still a density with respect to $\lambda^n$ (since $\int g\,d\lambda^n=1$) and satisfies
\begin{gather}\label{bv5rf7m}
g(x,y)=g(y,x)\quad\text{for all }x\in S^{n-1}\text{ and }y\in S.
\end{gather}
Next, define
\begin{gather*}
f_0(x)=\int g(x,v)\,\lambda^{n-1}(dv)\quad\quad\text{for all }x\in S,
\\f_{n-1}(x_n\mid x_1,\ldots,x_{n-1})=\frac{g(x)}{\int g(x_1,\ldots,x_{n-1},v)\,\lambda(dv)}
\end{gather*}
for all $x\in S^n$, and
$$f_{j-1}(x_j\mid x_1,\ldots,x_{j-1})=\frac{\int g(x,v)\,\lambda^{n-j}(dv)}{\int g(x_1,\ldots,x_{j-1},v)\,\lambda^{n-j+1}(dv)}$$
for all $2\le j\le n-1$ and $x\in S^j$. Finally, define a strategy $\sigma$ dominated by $\lambda$ as
\begin{gather*}
\sigma_0(dz)=f_0(z)\,\lambda(dz),
\\\sigma_j(x,\,dz)=f_j(z\mid x)\,\lambda(dz)
\end{gather*}
if $1\le j\le n-1$ and $x\in S^j$, and
\begin{gather*}\sigma_j(u,x)=\sigma_{n-1}(x)
\end{gather*}
if $j> n-1$, $u\in S^{j-n+1}$ and $x\in S^{n-1}$. Under $\sigma$, a density of $(X_1,\ldots,X_n)$ is given by $g$. By equation \eqref{bv5rf7m},
\begin{gather*}
\int g(v,x)\,\lambda(dv)=\int g(x,v)\,\lambda(dv)\quad\quad\text{for all }x\in S^{n-1}
\end{gather*}
and this in turn implies
\begin{gather*}
(X_2,\ldots,X_n)\overset{d}= (X_1,\ldots,X_{n-1}).
\end{gather*}
Therefore, $X$ is stationary because of Lemma \ref{f6yn9ik3e}. Note also that $X$ is a Markov sequence of order $n-1$.
\end{ex}

\section{Concluding remarks and open problems}

When prediction is the main target, P.A. has some advantages with respect to I.A. This is only our opinion, obviously, and we tried to support it along this paper. Even if one agrees, however, some further work is to be done to make P.A. a concrete tool. We close this paper with a brief list of open problems and possible hints for future research.

\medskip

\begin{itemize}

\item In various applications, the available information strictly includes the past observations on the variable to be predicted. For instance, as in Example \ref{h8uj92w}, suppose one aims to predict $X_{n+1}$ based on $(X_1,Z_1,\ldots,X_n,Z_n)$ where $Z_1,\ldots,Z_n$ are any random elements. Suppose also that $Z_1,\ldots,Z_n$ cannot be neglected for they are informative on $X_{n+1}$. In this case, one needs the conditional distribution of $X_{n+1}$ given $(X_1,Z_1,\ldots,X_n,Z_n)$. Situations of this type are practically meaningful and should be investigated further.

\medskip

\item Section \ref{x76t} should be expanded. It would be nice to have a general solution of Problem (*) for both the stationary and the stationary-ergodic cases. Further examples of stationary sequences (possibly, non-Markovian) would be welcome as well.

\medskip

\item Obviously, P.A. could be investigated under other distributional assumptions, in addition to exchangeability, stationarity and conditional identity in distribution. In particular, partial exchangeability should be taken into account.

\medskip

\item A question, related to Example \ref{f6y7ua1m}, is: Under what conditions $X$ is exchangeable when $\sigma$ is HMW's strategy ?

\medskip

\item While probably hard, the problem raised in Example \ref{3w34r6gc1} looks intriguing. In Theorem \ref{c5rt7n}, such a problem has been addressed when $\mu$ is a symmetric stable law and $f$ has a special form. What happens if $\mu$ and $f$ are arbitrary ?

\medskip

\item In case of I.A., the empirical Bayes point of view (where the prior is allowed to depend on the data) may be problematic. In case of P.A., instead, this point of view is certainly admissible. In fact, suppose a strategy $\sigma$ depends on some unknown constants, and an empirical Bayes forecaster decides to estimate these constants based on the available data. Acting in this way, she is merely replacing a strategy with another. Instead of $\sigma$, she is working with $\hat{\sigma}$, where $\hat{\sigma}$ is the strategy obtained from $\sigma$ estimating the unknown constants. This empirical form of P.A. looks reasonable and could be investigated.

\medskip

\end{itemize}

\begin{appendix}
\section*{Appendix}

This appendix contains the proofs of some claims scattered throughout the text. We will need the following characterization of c.i.d. sequences in terms of strategies.

\begin{thm}\textbf{(Theorem 3.1 of \cite{BPR2012}).}\label{l77be4c} Let $\sigma$ be a strategy. Then, $P_\sigma$ is c.i.d. if and only if
\begin{gather}\label{t6h8n}
\sigma_n(x,A)=\int\sigma_{n+1}(x,y,\,A)\,\sigma_n(x,\,dy)
\end{gather}
for all $n\ge 0$, all $A\in\mathcal{B}$ and $P_\sigma$-almost all $x\in S^n$.
\end{thm}

\medskip

\begin{proof}[Proof of Theorem \ref{7y7u8bf}] In this proof, ``density function'' stands for ``density function with respect to Lebesgue measure''. We first recall a well known fact.

\medskip

Let $C$ be a bivariate copula and $F_1$, $F_2$ distribution functions on $\mathbb{R}$. Suppose that $C$, $F_1$ and $F_2$ all have densities, say $c$, $f_1$ and $f_2$, respectively. Then,
\begin{gather*}
F(x,y)=C\bigl\{F_1(x),F_2(y)\bigr\}
\end{gather*}
is a distribution function on $\mathbb{R}^2$ and
\begin{gather*}
f(x,y)=c\bigl\{F_1(x),F_2(y)\bigr\}\,f_1(x)\,f_2(y)
\end{gather*}
is a density of $F$. Therefore, for all $y\in\mathbb{R}$ with $f_2(y)>0$, one obtains
\begin{gather*}
\int c\bigl\{F_1(x),F_2(y)\bigr\}\,f_1(x)\,dx=\int\frac{f(x,y)}{f_2(y)}\,dx=1.
\end{gather*}

\medskip

We next show that equation \eqref{g6n83k0u} actually defines a strategy $\sigma$. Fix a density $f_0>0$ and a sequence $c_1,c_2,\ldots$ of strictly positive bivariate copula densities. For each $y\in\mathbb{R}$,
\begin{gather*}
\int f_1(z\mid y)\,dz=\int c_1\bigl\{F_0(z),F_0(y)\bigr\}\,f_0(z)\,dz=1
\end{gather*}
since $f_0(y)>0$. Moreover, $f_1(z\mid y)>0$ for all $z$ due to $f_0>0$ and $c_1>0$. Next, suppose that $f_n(\cdot\mid x)$ is a strictly positive density for some $n\ge 1$ and $x\in\mathbb{R}^n$. Then, for all $y\in\mathbb{R}$,
\begin{gather*}
\int f_{n+1}(z\mid x,y)\,dz
\\=\int c_{n+1}\bigl\{F_n(z\mid x),F_n(y\mid x)\bigr\}\,f_n(z\mid x)\,dz=1
\end{gather*}
since $f_n(y\mid x)>0$. Furthermore, $f_{n+1}(z\mid x,y)>0$ for all $z$ since $f_n(\cdot\mid x)>0$ and $c_{n+1}>0$. By induction, this proves that $f_n(\cdot\mid x)$ is a density for all $n\ge 1$ and $x\in\mathbb{R}^n$. Therefore, equation \eqref{g6n83k0u} defines a strategy $\sigma$ (called HMW's strategy in Example \ref{f6y7ua1m}).

\medskip

Finally, we prove that $P_\sigma$ is c.i.d. if $\sigma$ is HMW's strategy. By Theorem \ref{l77be4c}, it suffices to prove condition \eqref{t6h8n}. In turn, since $\sigma$ is dominated by the Lebesgue measure, condition \eqref{t6h8n} reduces to
\begin{gather*}
f_n(z\mid x)=\int f_{n+1}(z\mid x,y)\,f_n(y\mid x)\,dy
\end{gather*}
for all $n\ge 0$, almost all $z\in\mathbb{R}$ and $P_\sigma$-almost all $x\in\mathbb{R}^n$. Such a condition follows directly from the definition of $\sigma$. In fact, for all $n\ge 0$ an $x\in\mathbb{R}^n$, one obtains
\begin{gather*}
\int f_{n+1}(z\mid x,y)\,f_n(y\mid x)\,dy
\\=\int c_{n+1}\bigl\{F_n(z\mid x),F_n(y\mid x)\bigr\}\,f_n(z\mid x)\,f_n(y\mid x)\,dy
\\=f_n(z\mid x)\quad\quad\text{for almost all }z.
\end{gather*}
This concludes the proof.
\end{proof}

\medskip

\begin{rem}\label{g73sd} HMW's strategy $\sigma$ has been defined under the assumption that $f_0>0$ and $c_n>0$ for all $n\ge 1$. Such an assumption is superfluous and has been made only to avoid annoying complications in the definition of $\sigma$. Similarly, $X$ is c.i.d. even if the $c_n$ are conditional copulas, in the sense that they are allowed to depend on past data. Precisely, for each $n\ge 1$ and $x\in\mathbb{R}^n$, fix a bivariate copula density $c_{n+1}(\cdot\mid x)$. Then, the proof Theorem \ref{7y7u8bf} still applies if $f_{n+1}(z\mid x,y)$ is rewritten as
\begin{gather*}
f_{n+1}(z\mid x,y)=c_{n+1}\bigl\{F_n(z\mid x),\,F_n(y\mid x)\mid x\bigr\}\,f_n(z\mid x).
\end{gather*}
\end{rem}

\bigskip

\begin{proof}[Proof of Theorem \ref{z34y9n1s}] We show that $X$ is c.i.d. via Theorem \ref{l77be4c}. Fix $A\in\mathcal{B}$ and $n\ge 0$. Since $P_\beta$ is exchangeable (and thus c.i.d.) Theorem \ref{l77be4c} yields
\begin{gather}\label{d3m8t6x}
\beta_n(x,A)=\int\beta_{n+1}(x,y,\,A)\,\beta_n(x,\,dy)
\end{gather}
for $P_\beta$-almost all $x\in S^n$. Hence, up to changing $\beta$ on a $P_\beta$-null set, equation \eqref{d3m8t6x} can be assumed to hold for all $x\in S^n$. If $n=0$,
\begin{gather*}
\int\sigma_1(y,A)\,\sigma_0(dy)=\int\beta_1(y,A)\,\beta_0(dy)=\beta_0(A)=\sigma_0(A)
\end{gather*}
where the first equality is because $\sigma_0=\beta_0$ and $\sigma_1=\beta_1$ while the second follows from \eqref{d3m8t6x}. Next, suppose $n\ge 1$ and take $x\in S^n$ and $y\in S$. By assumption, the events $\{T>n+1\}$ and $\{T\le n+1\}$ depend on $x$ but not on $y$. If $T>n+1$, one obtains $\sigma_{n+1}(x,y)=\beta_{n+1}(x,y)$ and $\sigma_n(x)=\beta_n(x)$. Hence, equation \eqref{d3m8t6x} implies again
\begin{gather*}
\int\sigma_{n+1}(x,y,\,A)\,\sigma_n(x,\,dy)=\int\beta_{n+1}(x,y,\,A)\,\beta_n(x,\,dy)
\\=\beta_n(x,A)=\sigma_n(x,A).
\end{gather*}
Similarly, if $T\le n+1$,
\begin{gather*}
\int\sigma_{n+1}(x,y,\,A)\,\sigma_n(x,\,dy)
\\=\int\bigl\{q_n(x)\,\sigma_n(x,A)\,+\,(1-q_n(x))\,\delta_y(A)\bigr\}\,\sigma_n(x,\,dy)
\\=q_n(x)\,\sigma_n(x,A)+(1-q_n(x))\,\int\delta_y(A)\,\sigma_n(x,\,dy)
\\=\sigma_n(x,A).
\end{gather*}
In view of Theorem \ref{l77be4c}, this proves that $X$ is c.i.d.

\medskip

Finally, suppose that $A_n$ is invariant under permutations of $S^n$ for each $n\ge 1$. We have to show that $(X_1,\ldots,X_n)$ is exchangeable conditionally on $T>n$. Fix $n$, a set $C\in\mathcal{B}^n$, and a permutation $\phi$ of $S^n$. For each $j\ge n$, it is easily seen that
\begin{gather*}
P_\sigma\Bigl(T=j+1,\,\,\phi(X_1,\ldots,X_n)\in C\Bigr)
\\=P_\beta\Bigl(T=j+1,\,\,\phi(X_1,\ldots,X_n)\in C\Bigr).
\end{gather*}
Therefore,
\begin{gather*}
P_\sigma\Bigl(T=j+1,\,\,\phi(X_1,\ldots,X_n)\in C\Bigr)
\\=P_\beta\Bigl(T=j+1,\,\,\phi(X_1,\ldots,X_n)\in C\Bigr)
\\=P_\beta\Bigl((X_1,\ldots,X_j)\in A_j,\,\,\phi(X_1,\ldots,X_n)\in C\Bigr)
\\=P_\beta\Bigl((X_1,\ldots,X_j)\in A_j,\,\,(X_1,\ldots,X_n)\in C\Bigr)
\end{gather*}
where the last equality is because $P_\beta$ is exchangeable and $A_j$ is invariant under permutations of $S^j$. In turn, this implies
\begin{gather*}
P_\sigma\Bigl(T>n,\,\,\phi(X_1,\ldots,X_n)\in C\Bigr)
\\=\sum_{j\ge n}P_\sigma\Bigl(T=j+1,\,\,\phi(X_1,\ldots,X_n)\in C\Bigr)
\\=\sum_{j\ge n}P_\beta\Bigl(T=j+1,\,\,(X_1,\ldots,X_n)\in C\Bigr)
\\=\sum_{j\ge n}P_\sigma\Bigl(T=j+1,\,\,(X_1,\ldots,X_n)\in C\Bigr)
\\=P_\sigma\Bigl(T>n,\,\,(X_1,\ldots,X_n)\in C\Bigr).
\end{gather*}
This concludes the proof.
\end{proof}

\bigskip

\begin{proof}[Proof of Theorem \ref{g78j92xcd1q}]
Just note that $g_n$ is a density of $(X_1,\ldots,X_n)$ with respect to $\lambda^n$. Therefore, Theorem \ref{g78j92xcd1q} follows from the very definitions of stationarity and exchangeability, after noting that $\int g_{n+1}(u,\cdot)\,\lambda(du)$ is a density of $(X_2,\ldots,X_{n+1})$ with respect to $\lambda^n$.
\end{proof}

\bigskip

\begin{proof}[Proof of Theorem \ref{c5rt7n}]
We first recall that
\begin{gather*}
\int\mathcal{S}(x,b;\,A)\,\mathcal{S}(0,r;\,dx)=\mathcal{S}(0,b+r;\,A)
\end{gather*}
for all $A\in\mathcal{B}$ and $b,\,r>0$. This can be checked by a direct calculation. For a proof, we refer to the Claim of \cite[Th. 10]{BPCID}. Having noted this fact, define
\begin{gather*}
\mu=\mathcal{S}(a,b),\quad f(x)=-a+c\,x,\quad\nu=\mathcal{S}\left(0,\,\frac{b}{1-\abs{c}^\gamma}\right),
\end{gather*}
and denote by $Z$ a real random variable such that $Z\overset{d}=\mathcal{S}(0,1)$. Define also
\begin{gather*}
r=\frac{b\,\abs{c}^\gamma}{1-\abs{c}^\gamma},\quad h(x)=c\,x,
\end{gather*}
and call $\nu^*$ the probability distribution of $h$ under $\nu$. On noting that
\begin{gather*}
a+b^{1/\gamma}Z\overset{d}=\mu\quad\text{and}\quad\nu^*=\mathcal{S}(0,\,r),
\end{gather*}
one obtains
\begin{gather*}
\int\sigma_1(x,A)\,\nu(dx)=\int P\bigl(f(x)+a+b^{1/\gamma}Z\in A\bigr)\,\nu(dx)
\\=\int P\bigl(h(x)+b^{1/\gamma}Z\in A\bigr)\,\nu(dx)
\\=\int P\bigl(x+b^{1/\gamma}Z\in A\bigr)\,\nu^*(dx)
\\=\int\mathcal{S}(x,b;\,A)\,\mathcal{S}(0,r;\,dx)=\mathcal{S}(0,\,b+r;\,A)
\\=\mathcal{S}\left(0,\,\frac{b}{1-\abs{c}^\gamma};\,A\right)=\nu(A).
\end{gather*}
Therefore, equation \eqref{g67y84f} holds.
\end{proof}

\end{appendix}

\medskip

\textbf{Acknowledgments:} We are grateful to Federico Bassetti and Paola Bortot for very useful conversations.

\end{document}